\documentclass[a4paper,aps,pra,showpacs,superscriptaddress,twocolumn,nofootinbib]{revtex4-1}      
\usepackage[utf8]{inputenc}  
\usepackage[T1]{fontenc}     
\usepackage[british]{babel}  
\usepackage[sc,osf]{mathpazo}
\usepackage[scaled=0.86]{berasans}  
\usepackage[scaled=1.03]{inconsolata} 
\usepackage[colorlinks,linkcolor=magenta,urlcolor=blue]{hyperref}  
\usepackage{graphicx} 
\usepackage[babel]{microtype}  
\usepackage{amsmath,amssymb,amsthm,bm,mathtools,amsfonts,mathrsfs,bbm} 
\usepackage{xspace}  
\usepackage{pgfplots}  

\DeclareMathOperator{\tr}{tr}

\newcommand{\ket}[1]{\left| #1 \right\rangle}

\newcommand{\ketbra}[2]{\left|#1\middle\rangle\middle\langle#2\right|}

\newcommand{\mean}[1]{\left\langle#1\right\rangle}

\newcommand{\ie}{\emph{i.e.}\@\xspace}

\newtheorem{theorem}{Theorem}

\def\be{\begin{equation}}
\def\ee{\end{equation}}
\mathtoolsset{centercolon}

\begin{document}
\title{Joint measurability, Einstein-Podolsky-Rosen steering, and Bell nonlocality}
\author{Marco Túlio Quintino}
\affiliation{Département de Physique Théorique, Université de Genève, 1211 Genève, Switzerland}
\author{Tamás V\'ertesi}
\affiliation{Institute for Nuclear Research, Hungarian Academy of Sciences, H-4001 Debrecen, P.O. Box 51, Hungary}
\author{Nicolas Brunner}
\affiliation{Département de Physique Théorique, Université de Genève, 1211 Genève, Switzerland}

\date{\today}  

\begin{abstract}
We investigate the relation between the incompatibility of quantum measurements and quantum nonlocality. 
We show that a set of measurements is not jointly measurable (\ie incompatible) if and only if it can be used for demonstrating Einstein-Podolsky-Rosen steering, a form of quantum nonlocality. Moreover, we discuss the connection between Bell nonlocality and joint measurability, and give evidence that both notions are inequivalent. Specifically, we exhibit a set of incompatible quantum measurements and show that it does not violate a large class of Bell inequalities. This suggest the existence of incompatible quantum measurements which are Bell local, similarly to certain entangled states which admit a local hidden variable model.
\end{abstract}

\maketitle

The correlations resulting from local measurements on an entangled quantum state cannot be explained by a local theory. This aspect of entanglement, termed quantum nonlocality, is captured by two inequivalent notions, namely Bell nonlocality \cite{bell64,brunner_review} and EPR steering \cite{schrodinger36,wiseman07,reid09}. The strongest form of this phenomenon is Bell nonlocality, witnessed via the violation of Bell inequalities. Steering represents a strictly weaker form of quantum nonlocality \cite{wiseman07}, witnessed via violation of steering inequalities \cite{cavalcanti09}. Both aspects have been extensively investigated in recent years, as they play a central role in the foundations of quantum theory and in quantum information processing.

Interestingly quantum nonlocality is based on two central features of quantum theory, namely entanglement and incompatible measurements. Specifically, performing (i) arbitrary local measurements on a separable state, or (ii) compatible measurements on an (arbitrary) quantum state can never lead to any form of quantum nonlocality. Hence the observation of quantum nonlocality implies the presence of both entanglement and incompatible measurements. It is interesting to explore the converse problem. Two types of questions can be asked here (see Fig. \ref{Fig1}): (a) do all entangled states lead to quantum nonlocality? (b) do all sets of incompatible measurements lead to quantum nonlocality?

An intense research effort has been devoted to question (a). First, it was shown that all pure entangled states violate a Bell inequality \cite{gisin91,popescu92}, hence also demonstrating EPR steering. For mixed states, the situation is much more complicated. There exist entangled states which are local, in the sense that no form of quantum nonlocality can be demonstrated with such states when using non-sequential measurements \cite{werner89,barrett02}. These issues become even more subtle when more sophisticated measurement scenarios are considered \cite{popescu95,hirsch13,palazuelos12,cavalcanti12}.

Question (b) has received much less attention so far. In the case of projective measurements, it was shown that incompatible measurements can always lead to Bell nonlocality \cite{tsirelson85,wolf09}. Note that in this case, compatibility is uniquely captured by the notion of commutativity \cite{varabook}. However, for general measurements, \ie positive-operator-valued-measures (POVMs), no general result
is known. In this case, there are several inequivalent notions of compatibility. Here we focus on the notion of joint measurability, see e.g.\cite{buschbook}, as this represents a natural choice in the context of quantum nonlocality. Several works discussed question (b) for POVMs \cite{son05,anderson05}. The strongest result is due to Wolf et al. \cite{wolf09}, who showed that any set of two incompatible POVMs with binary outcomes can always lead to violation of the Clauser-Horne-Shimony-Holt Bell inequality. However, this result may not be extended to the general case (of an arbitrary number of POVMs with arbitrarily many outcomes), since pairwise joint measurability does not imply full joint measurability in general \cite{krausbook}.

\begin{figure}[b!]
\includegraphics[width = \columnwidth]{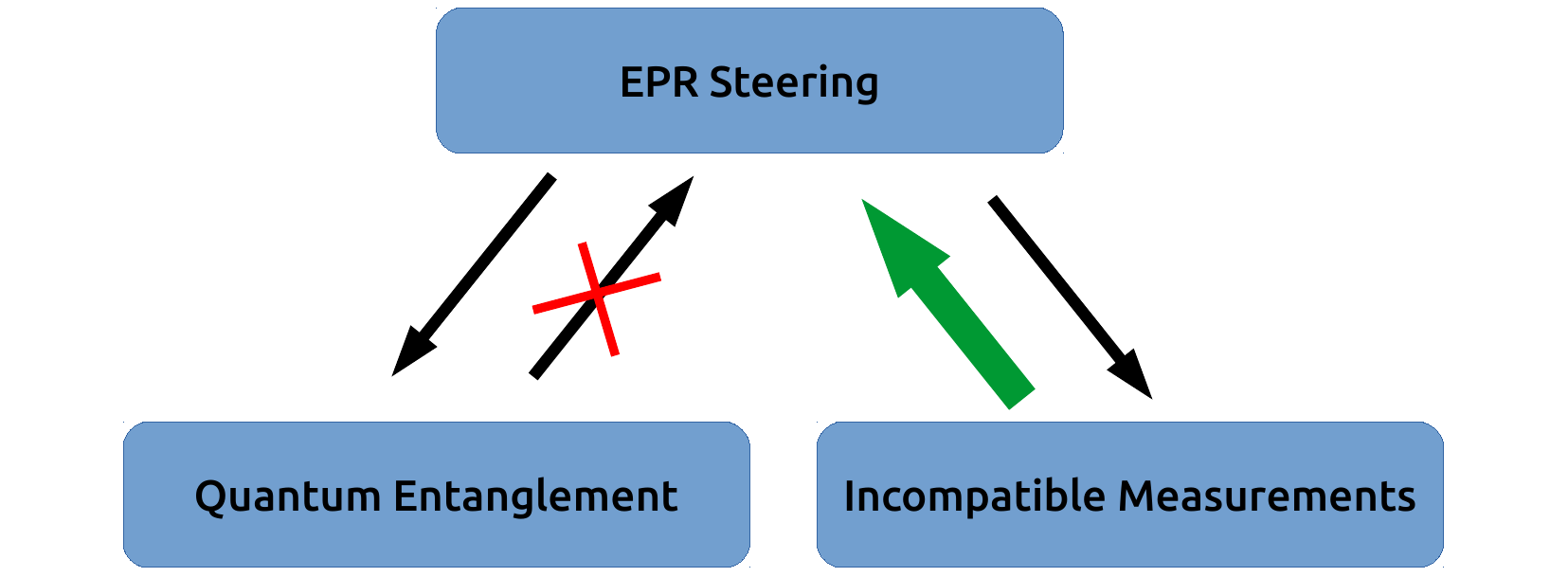}
    \caption{The observation of EPR steering, a form of quantum nonlocality, implies the presence of both entanglement and incompatible measurements. Whether the converse links hold is an interesting question. Here we make progress in this direction by showing that any set of incompatible measurements can be used to demonstrate EPR steering (green arrow).}
\label{Fig1}
\end{figure}

Here we explore the relation between compatibility of general quantum measurements and quantum nonlocality. We start by demonstrating a direct link between joint measurability and EPR steering. Specifically, we show that for any set of POVMs that is incompatible (\ie not jointly measurable), one can find an entangled state, such that the resulting statistics violates a steering inequality. Hence the use of incompatible is a necessary and sufficient ingredient for demonstrating EPR steering.

This raises the question of how joint measurability relates to Bell nonlocality. Specifically, the question is whether, for any set of incompatible POVMs (for Alice), one can find an entangled state and a set of local measurements (for Bob), such that the resulting statistics violates a Bell inequality. Here we give evidence that the answer is negative. In particular, we exhibit sets of incompatible measurements which can provably not violate a large class of Bell inequalities (including all full correlation Bell inequalities, also known as XOR games, see \cite{brunner_review}). We therefore conjecture that non joint measurability and Bell nonlocality are inequivalent. Hence, similarly to local entangled states, there may exist incompatible quantum measurements which are Bell local.

\emph{Steering vs joint measurability.} We start by defining the relevant scenario and notations. We consider two separated observers, Alice and Bob, performing local measurements on a shared quantum state $\rho_{AB}$. Alice's measurements are represented by operators $M_{a|x}$ such that $\sum_a M_{a|x} = \openone$, where $x$ denotes the choice of measurement and $a$ its outcome. Upon performing measurement $x$, and obtaining outcome $a$, the (unnormalized) state held by Bob is given by
\begin{equation}
\sigma_{a\vert x}  = \tr_A(\rho_{AB} M_{a\vert x} \otimes \openone ).
\end{equation}
The set of unnormalised states $\{\sigma_{a\vert x} \}$, referred to as an \textit{assemblage}, completely characterizes the experiment, since
$\tr(\sigma_{a\vert x})$ is the probability of Alice getting the output $a$ (for measurement $x$) and given that information and Bob's state is described by ${\sigma_{a\vert x}}/{\tr(\sigma_{a\vert x})}$. Importantly,
one has that $\sum_a \sigma_{a\vert x} = \sum_a \sigma_{a\vert x'}$ for all measurements ${x}$ and ${x}'$, ensuring that Alice cannot signal to Bob.

In a steering test \cite{wiseman07}, Alice want to convince Bob that the state $\rho_{AB}$ is entangled, and that she can steer his state. Bob does not trust Alice, and thus wants to verify Alice's claim. Asking Alice to perform a given measurement $x$, and to announce the outcome $a$, Bob can determine the assemblage ${\sigma_{a\vert x}}$ via local quantum tomography. To ensure that steering did indeed occur, Bob should verify that the assemblage does not admit a decomposition of the form
\begin{equation} \label{LHS}
    \sigma_{a\vert x}= \sum_\lambda \pi(\lambda) p(a\vert x,\lambda) \sigma_\lambda,
\end{equation}
where $\sum_\lambda \pi(\lambda)=1$. Clearly, if a decomposition of the above form exists, then Alice could have cheated by sending the (unentangled) state $\sigma_\lambda$ to Bob and announce outcome $a$ to Bob according to the distribution $p(a\vert x,\lambda)$. Note that here $\lambda$ represents a local variable of Alice, representing her choice of strategy.

Assemblages of the form \eqref{LHS} are termed 'unsteerable' and form a convex set \cite{pusey13,skrzypczyk14}. Hence any 'steerable' assemblage can be detected via a set of linear witnesses called steering inequalities \cite{cavalcanti09}. By observing violation of a steering inequality, Bob will therefore be convinced that Alice can steer his state.

For a demonstration of steering, it is necessary for the state $\rho_{AB}$ to be entangled. However, not all entangled states can be used to demonstrate steering \cite{wiseman07,barrett02,bowles14}; at least not when non-sequential measurements are performed on a single copy of $\rho_{AB}$.

Moreover, steering also requires that the measurements performed by Alice are incompatible. To capture the compatibility of a set of quantum measurements we use here the notion of joint measurability, see e.g. \cite{buschbook}. A set of $m$ POVMs $M_{a|x}$ is called jointly measurable if there exists a measurement $M_{\vec{a}}$ with outcome $\vec{a}=[a_{x=1},a_{x=2},\ldots,a_{x=m}]$ where $a_{x}$ gives the outcome of measurement $x$, that is
\begin{align}
         M_{\vec{a}} \geq  0 , \quad \sum_{\vec{a}} M_{\vec{a}} =  \openone , \quad
        \sum_{\vec{a}\setminus a_x}M_{\vec{a}} = M_{a\vert x}  \;,
\end{align}
where $\vec{a}\setminus a_x$ stands for the elements of $\vec{a}$ except for $a_x$. Hence, all POVM elements $M_{a|x}$ are recovered as marginals of the \textit{Mother Observable} $M_{\vec{a}}$. Importantly, the joint measurability of a set of POVMs does not imply that they commute \cite{kru}. Hence joint measurability is a strictly weaker notion of compatibility for POVMs. Moreover, joint measurability is not transitive. For instance, pairwise joint measurability does not imply full joint measurability in general \cite{krausbook} (see below).

Our main result is to establish a direct link between joint measurability and steering. Specifically, we show that a set of POVMs can be used to demonstrate steering if and only if it is not jointly measurable. More formally we prove the following result.

\begin{theorem}
    The assemblage $\{ \sigma_{a|x} \}$, with $\sigma_{a|x} = \tr_A(\rho_{AB} M_{a\vert x} \otimes \openone )$, is unsteerable for any state $\rho_{AB}$ acting in $ \mathbb{C}^d \otimes \mathbb{C}^d$ if and only if the set of POVMs $\{ M_{a\vert x} \}$ acting on $\mathbb{C}^d$ is jointly measurable.
\end{theorem}

\begin{proof}
The 'if' part is straightforward. Our goal is to show that $\{ \sigma_{a|x} \}$ admits a decomposition of the form \eqref{LHS} when $\{ M_{a\vert x} \}$ is jointly measurable, for any state $\rho_{AB}$. Consider $M_{\vec{a}}$, the mother observable for $\{ M_{a\vert x} \}$, and define Alice's local variable to be $\lambda = \vec{a} $, distributed according to $\Pi(\vec{a}) = \tr(M_{\vec{a}} \rho_A)$, where $\rho_A = \tr_B(\rho_{AB})$. Next Alice sends the local state $\sigma_{\vec{a}} = \tr_A(M_{\vec{a}} \otimes \openone \rho_{AB}) /  \Pi(\vec{a})$. When asked by Bob to perform measurement $x$, Alice announces an outcome $a$ according to $p(a|x,\vec{a}) = \delta_{a,a_x}$.

We now move to the 'only if' part. Consider an arbitrary pure state $\rho_{AB}=\ketbra{\psi}{\psi}$ with Schmidt number $d$. Notice that we can always write $\ket{\psi}=(D \otimes \openone) \ket{\Phi}$, where $\ket{\Phi}=\sum_i \ket{ii}$ is an (unormalized) maximally entangled state in $\mathbb{C}^d \otimes \mathbb{C}^d$, and $D$ is diagonal matrix that contains only strictly positive numbers. The assemblage resulting from a set of POVMs $\{ M_{a\vert x} \}$ on $\rho_{AB}$ is given by
\begin{equation}
\sigma_{a\vert x}= \tr_A(M_{a|x} \otimes \openone \ketbra{\psi}{\psi}) = D M_{a\vert x}^T D
\end{equation}
where $M_{a\vert x}^T$ is the transpose of $M_{a\vert x}$. Our goal is now to show that if $\sigma_{a\vert x}$ is unsteerable then $\{ M_{a\vert x} \}$ is jointly measurable. As $\sigma_{a\vert x}$ is unsteerable, we have that 
\begin{equation} \sigma_{a\vert x}=\sum_\lambda \pi(\lambda) p(a\vert x,\lambda) \sigma_\lambda, 
\end{equation}
which allows us to define the positive definite operator
\begin{equation}
\sigma_{\vec{a}} =\sum_\lambda \pi(\lambda) \sigma_\lambda \prod_x p(a_x \vert x,\lambda)
\end{equation}
form which we can recover the assemblage $\{\sigma_{a\vert x}\}$ as marginals, \ie $\sigma_{a\vert x}=\sum_{\vec{a}\setminus a_x} \sigma_{\vec{a}}$. Since the diagonal matrix $D$ is invertible, we can define $M_{\vec{a}}:= D^{-1}\sigma_{\vec{a}}^T D^{-1}$. It is straightforward to check that $M_{\vec{a}}$ is a mother observable for $\{M_{a\vert x}\}$: (i) it is positive, (ii) sums to identity, and (iii) has POVM elements $M_{a\vert x}$ as marginals. Hence $\{M_{a\vert x}\}$ is jointly measurable, which concludes the proof. Note finally an interesting point that follows from the above. Considering a set of incompatible measurements acting on $\mathbb{C}^d$, any pure entangled state of the Schmidt number $d$ can be used to demonstrate EPR steering. 
\end{proof}

\emph{Bell nonlocality vs joint measurability.} It is natural to ask whether the above connection, between joint measurability and steering, can be extended to Bell nonlocality. Recall that in a Bell test, both observers Alice and Bob are on the same footing, and test the strength of the shared correlations. Specifically, Alice chooses a measurement $x$ (Bob chooses $y$) and gets outcome $a$ (Bob gets $b$). The correlation is thus described by a joint probability distribution $p(ab|xy)$. The latter can be reproduced by a pre-determined classical strategy if it admits a decomposition of the form
\begin{equation} \label{LHV}
    p(ab \vert xy)= \sum_\lambda \pi(\lambda) p(a\vert x,\lambda) p(b \vert y,\lambda).
\end{equation}
where $\lambda$ represents the shared local (hidden) variable, and $ \sum_\lambda \pi(\lambda)=1$. Any distribution that does not admit a decomposition of the above form is said to be Bell nonlocal. The set of local distributions, \ie of the form \eqref{LHV} is convex, and can thus be characterized by a set of linear inequalities called Bell inequalities \cite{brunner_review}. Hence violation of a Bell inequality implies Bell nonlocality.

In quantum theory, Bell nonlocal distributions can be obtained by performing suitably chosen local measurements, $M_{a|x}$ and $M_{b|y}$, on an entangled state $\rho_{AB}$. In this case, the resulting distribution $p(ab|xy) = \tr(\rho_{AB} M_{a|x} \otimes M_{b|y})$ does not admit a decomposition of the form \eqref{LHV}. Bell nonlocality is however not a generic feature of entangled quantum states. That is, there exist mixed entangled states which are local, in the sense that the statistics resulting from arbitrary non-sequential local measurements can be reproduced by a local model \cite{werner89,barrett02,hirsch13}.

Given the above, we investigate now how joint measurability relates to Bell nonlocality. First the above theorem implies that, if the set of POVMs $\{ M_{a|x} \}$ used by Alice is jointly measurable, then the statistics $p(ab|xy)$ can always be reproduced by a local model, for any state $\rho_{AB}$ and measurements of Bob $\{ M_{b|y} \}$. The converse problem is much more interesting. The question is whether for any set of POVMs $\{ M_{a|x} \}$ that is not jointly measurable, there exists a state $\rho_{AB}$ and a set of measurements $\{ M_{b|y} \}$ such that the resulting statistics $p(ab|xy)$ violates a Bell inequality. This was shown to hold true for the case of sets of two POVMs with binary outcomes \cite{wolf09}. In this case, joint measurability is equivalent to violation of the CHSH Bell inequality. Here we give evidence that this connection does not hold in general. Specifically, we exhibit a set of POVMs which is not jointly measurable but nevertheless cannot violate a large class of Bell inequalities.

Consider the set of three dichotomic POVMs (acting on $\mathbb{C}^2$) given by the following positive operators
\begin{align} \label{hollow_triangle}
M_{0|x}'(\eta) = \frac{1}{2} \left( \openone + \eta \sigma_x \right)
\end{align}
for $x=1,2,3$, where $\sigma_1,\sigma_2,\sigma_3$ are the Pauli matrices, and $0\leq \eta\leq 1$. Indeed, one has that $M_{1|x}(\eta)'= \openone - M_{0|x}'(\eta)$. This set of POVMs should be understood as noisy Pauli measurements. The set is jointly measurable if and only if $\eta\leq 1/\sqrt{3}$, although any pair of POVMs is jointly measurable for $\eta\leq 1/\sqrt{2}$ \cite{teiko08} (see also \cite{YC}). Hence in the range $1/\sqrt{3}\leq \eta\leq 1/\sqrt{2}$, the set $\{ M_{a|x}'(\eta) \}$ forms a \textit{hollow triangle}: it is pairwise jointly measurable but not fully jointly measurable.

We now investigate whether the above hollow triangle can lead to
Bell inequality violation. The most general class of Bell
inequalities to be considered here are of the form~\footnote{Note
that in principle Bob may use POVMs. Since we can restrict to pure
two-qubit entangled states, it is sufficient to consider for Bob
POVM with 4 outcomes: see \cite{dariano05}.}:
\begin{equation} \label{BI}
  I = \sum_{x=1}^3 \sum_{y=1}^n \gamma_{xy} \mean{A_xB_y}+\sum_{x=1}^3 \alpha_{x} \mean{A_x}+ \sum_{y=1}^n \beta_{y} \mean{B_y}\leq 1
\end{equation}
where
\begin{align} \label{corrs}
    &\mean{A_xB_y}= p(a=b\vert xy)-p(a\neq b\vert xy); \\
     &\mean{A_x}= p(0\vert x)-p(1\vert x), \; \mean{B_y}= p(0\vert y)-p(1\vert y) \nonumber .
\end{align}
All (tight) Bell inequalities of the above form for $n\leq 5$ are known (see Appendix). Using a numerical method based on semi-definite-programming (SDP) \cite{navascues12} (see Appendix) we could find the smallest value of the parameter $\eta$ for which a given inequality can be violated using the set of POVMs \eqref{hollow_triangle}. The results are summarized in Table I. Notably, we could not find a violation in the range $1/\sqrt{3}\leq \eta\leq 1/\sqrt{2}$ where the set $\{ M_{a|x}'(\eta) \}$ is a hollow triangle. In fact, no violation was found for $\eta\leq 0.7786$, whereas pairwise joint measurability is achieved for $\eta\leq 1/\sqrt{2} \simeq 0.7071$, thus leaving a large gap. Note also that pairwise joint measurability implies violation of the CHSH inequality here, since we have POVMs with binary outcomes \cite{wolf09}. We thus conjecture that there is a threshold value $\eta^*>1/\sqrt{3}$, such that all hollow triangles with $ 1/\sqrt{3} < \eta \leq \eta^*$ do not violate any Bell inequality.
\begin{table}[t!] \label{table}
\begin{tabular}{|c||c|c|}
\hline
& $\{ M_{a|x}'(\eta) \}$ & $ \{ M_{a|x}''(\eta) \} $ \\
\hline \hline
Pairwise JM (CHSH violation) & $1/\sqrt{2}\approx 0.7071$ & 0.5858 \\
\hline
Triplewise JM & $1/\sqrt{3}\approx 0.5774$ & 0.4226 \\
\hline \hline
Bell violation ($n=3$): $I_{3322}$ & 0.8037 & 0.6635 \\
\hline
Bell violation ($n=4$): $I^1_{3422}$ & 0.8522 & 0.7913 \\
$I^2_{3422}$ & 0.8323 & 0.5636 \\
$I^3_{3422}$ & 0.8188 & 0.6795 \\
\hline
Bell violation ($n=5$): $I_{3522}$ & 0.7786 & 0.5636 \\
\hline
\end{tabular}
\caption{ Bell inequality violation with incompatible POVMs. Specifically, we consider the sets given in equations \eqref{hollow_triangle} and \eqref{HT2}. For each set, we determine the smallest value of the parameter $\eta$, such that the set becomes Jointly Measurable (JM), and achieve Bell inequality violation. We consider tight Bell inequalities with up to $n=5$ measurements for Bob (see Appendix). Note that pairwise joint measurability is equivalent to violation of the CHSH Bell inequality.}

\end{table}

Moreover, we can also show that a large class of Bell inequalities of the form \eqref{BI} (for arbitrary $n$) cannot be violated using the hollow triangle \eqref{hollow_triangle}. Note that for the set of POVMs \eqref{hollow_triangle}, we have that $\mean{A_x} = \eta \tr(\sigma_x \rho_A)$ and $\mean{A_xB_y}=  \eta  \tr(\sigma_x  \otimes M_{b|y} \rho_{AB})$ for $x=1,2,3$. Hence we can write the Bell polynomial as
\begin{equation}
I = \eta \tilde{I} + (1-\eta) \sum_{y} \beta_{y} \mean{B_y}
\end{equation}
where $\tilde{I}$ is the Bell expression $I$ evaluated for projective (Pauli) measurements on Alice's side. Note that
$\tilde{I} \leq I_{\mathbb{C}^2}$, where $I_{\mathbb{C}^2}$ denotes the maximal value of $I$
 for qubit strategies. Hence no Bell inequality violation is possible when
\begin{equation}
\sum_y |\beta_y| \leq \frac{1- \eta I_{\mathbb{C}^2} }{1-\eta},
\end{equation}
given that $\eta I_{\mathbb{C}^2} \leq 1$.
 Notably this includes all full correlation Bell inequalities ($\alpha_x=\beta_y=0$), \ie XOR games, for which it is known that the amount of violation is upper bounded for qubit strategies. More precisely, one has that $I_{\mathbb{C}^2} \leq K_3$ \cite{tsirelson93,tamas} where $K_3 \leq 1.5163$ is the Grothendieck constant of order 3. Hence, for $  1/ \sqrt{3} < \eta < 1/K_3\simeq 0.6595 $ we get that the hollow triangle \eqref{hollow_triangle} cannot violate any full correlation Bell inequality.

From the above, one may actually wonder whether Bell inequality violation is possible at all using set of POVMs forming a hollow triangle. We now show that this is the case. Consider the set of three dichotomic POVMs (acting on $\mathbb{C}^2$) given by the following positive operators
\begin{align} \label{HT2}
M_{0|x}''(\eta) =  \frac{\eta}{2} \left( \openone + \sigma_x \right)
\end{align}
for $x=1,2,3$ and $0\leq \eta\leq 1$. Again, one has that
$M_{1|x}''(\eta)= \openone - M_{0|x}''(\eta)$. To determine the
range of the parameter $\eta$ for which the above set of POVMs is
pairwise jointly measurable, and fully jointly measurable, we use
the SDP techniques of Ref. \cite{wolf09}. We find that the set $\{
M_{a|x}''(\eta) \}$ is a hollow triangle for $0.4226 \leq \eta
\leq 0.5858$. However, Bell nonlocality can be obtained by
considering a Bell inequality with $n=4$ measurements for Bob, for
$\eta>0.5636$. Values are summarized in Table I, while details of
the construction are given in the Appendix. This shows that a set
of partially compatible measurements, here a hollow triangle, can
be used to violate a Bell inequality. Moreover, this suggests that
detecting the nonlocality of a set of 3 incompatible POVMs is a
hard problem, since a large number of measurements on Bob's side
(possibly infinite) might be needed. This contrasts with the case
of two POVMs, where two measurements (via CHSH) were enough
\cite{wolf09}. Finally, note that we could find a hollow triangle
with only real numbers (\ie with all Bloch vectors in a plane of
the sphere) which violates a Bell inequality with $n=3$
measurements \cite{collins04}, \ie the simplest possible case (see Appendix).

Finally, an interesting open question is the following. Considering a set of arbitrarily many POVMs, it is known that any partial compatibility configuration can be realized \cite{fritz14}. Is it then possible to violate a Bell inequality for any possible configuration?

\emph{Discussion.} We have discussed the relation between joint measurability and quantum nonlocality. First, we showed that a set of  POVMs is incompatible if and only if it can be used to demonstrate EPR steering. Hence, EPR steering provides a new operational interpretation of joint measurability. Second, we explored the link between joint measurability and Bell nonlocality. We gave evidence that these two notions are inequivalent, by showing that a hollow triangle (a set of 3 POVMs that is only pairwise compatible) can never lead to violation of a large class of Bell inequalities. We conjecture that this hollow triangle is Bell local, that is, it cannot be used to violate any Bell inequality. Hence such a measurement would represent the analogue, for a quantum measurement, of a local entangled state.

\emph{Acknowledgements.}
We thank J. Bowles, D. Cavalcanti, R. Chaves, F. Hirsch, M. Navascues, S. Pironio and P. Skrzypczyk for discussions, and D. Rosset for help with the website faacets.com. We acknowledge financial support from the Swiss National Science Foundation (grant PP00P2\_138917 and QSIT), the SEFRI (COST action MP 1006) and the EU DIQIP, the J\'anos Bolyai Programme, the OTKA (PD101461) and the T\'AMOP-4.2.2.C-11/1/KONV-2012-0001 project.

\emph{Note added.} While the present work was under review, we
became aware of a related work \cite{Uola}.



\begin{thebibliography}{10}

\bibitem{bell64}
J.~S. Bell, {\em Physics}
  {\bfseries 1}, 195--200 (1964).

\bibitem{brunner_review}
N.~{Brunner}, D.~{Cavalcanti}, S.~{Pironio}, V.~{Scarani}, and S.~{Wehner},
  \href{http://dx.doi.org/10.1103/RevModPhys.86.419}{{\em Rev. Mod. Phys.} {\bfseries 86}, 419--478 (2014)}.

\bibitem{schrodinger36}
E.~Schrödinger, 
  \href{http://dx.doi.org/10.1017/S0305004100019137}{{\em Mathematical
  Proceedings of the Cambridge Philosophical Society} {\bfseries 32}, 446--452
  (1936)}.

\bibitem{wiseman07}
H.~M. Wiseman, S.~J. Jones, and A.~C. Doherty, 
  \href{http://dx.doi.org/10.1103/PhysRevLett.98.140402}{{\em Phys. Rev. Lett.}
  {\bfseries 98}, 140402 (2007)}.

\bibitem{reid09}
M.~D. Reid, P.~D. Drummond, W.~P. Bowen, E.~G. Cavalcanti, P.~K. Lam, H.~A.
  Bachor, U.~L. Andersen, and G.~Leuchs, 
  \href{http://dx.doi.org/10.1103/RevModPhys.81.1727}{{\em Rev. Mod. Phys.}
  {\bfseries 81}, 1727--1751 (2009)}.

\bibitem{cavalcanti09}
E.~G. Cavalcanti, S.~J. Jones, H.~M. Wiseman, and M.~D. Reid, 
  \href{http://dx.doi.org/10.1103/PhysRevA.80.032112}{{\em Phys. Rev. A}
  {\bfseries 80}, 032112 (2009)}.

\bibitem{gisin91}
N.~Gisin, 
  \href{http://dx.doi.org/http://dx.doi.org/10.1016/0375-9601(91)90805-I}{{\em
  Phys. Lett. A} {\bfseries 154}, 201 -- 202 (1991)}.

\bibitem{popescu92}
S.~Popescu and D.~Rohrlich, 
  \href{http://dx.doi.org/http://dx.doi.org/10.1016/0375-9601(92)90711-T}{{\em
  Phys. Lett. A} {\bfseries 166}, 293 -- 297 (1992)}.


\bibitem{werner89}
R.~F. Werner, 
  \href{http://dx.doi.org/10.1103/PhysRevA.40.4277}{{\em Phys. Rev. A}
  {\bfseries 40}, 4277--4281 (1989)}.

\bibitem{barrett02}
J.~{Barrett}, 
  \href{http://dx.doi.org/10.1103/PhysRevA.65.042302}{{\em Phys. Rev.~A}
  {\bfseries 65}, 042302 (2002)}.

\bibitem{popescu95}
S.~{Popescu}, 
 \href{http://dx.doi.org/10.1103/PhysRevLett.74.2619}{{\em
  Phys. Rev. Lett.} {\bfseries 74}, 2619--2622 (1995)}.

\bibitem{hirsch13}
F.~{Hirsch}, M.~T. {Quintino}, J.~{Bowles}, and N.~{Brunner}, 
  \href{http://dx.doi.org/10.1103/PhysRevLett.111.160402}{{\em Phys. Rev.
  Lett.} {\bfseries 111}, 160402 (2013)}.

\bibitem{palazuelos12}
C.~{Palazuelos}, 
  \href{http://dx.doi.org/10.1103/PhysRevLett.109.190401}{{\em Phys. Rev.
  Lett.} {\bfseries 109}, 190401 (2012)}.

\bibitem{cavalcanti12}
D.~{Cavalcanti}, A.~{Ac{\'{\i}}n}, N.~{Brunner}, and T.~{V{\'e}rtesi},
  \href{http://dx.doi.org/10.1103/PhysRevA.87.042104}{{\em {Phys. Rev.~A}}
  {\bfseries 87}, 042104 (2013)}.

\bibitem{tsirelson85}
L.~A. Khalfin and B.~S. Tsirelson, {\em Symposium on the Foundations of Modern Physics},
 441-460 (1985).

\bibitem{wolf09}
M.~M. {Wolf}, D.~{Perez-Garcia}, and C.~{Fernandez}, 
  \href{http://dx.doi.org/10.1103/PhysRevLett.103.230402}{{\em Phys. Rev.
  Lett.} {\bfseries 103}, 230402 (2009)}.
  
\bibitem{varabook} V.S. Varadarajan, {\em Geometry of Quantum Theory}, Springer 1985.

\bibitem{buschbook}
P.~Busch, P.~Lahti, and P.~Mittelstaedt, {\em The Quantum Theory of
  Measurement}.
\newblock Lecture Notes in Physics Monographs Vol. 2, Springer 1996, pp 25-90.

\bibitem{son05}
W.~{Son}, E.~{Andersson}, S.~M. {Barnett}, and M.~S. {Kim}, 
  \href{http://dx.doi.org/10.1103/PhysRevA.72.052116}{{\em Phys. Rev.~A}
  {\bfseries 72}, 052116 (2005)}.

\bibitem{anderson05}
E.~{Andersson}, S.~M. {Barnett}, and A.~{Aspect}, 
  \href{http://dx.doi.org/10.1103/PhysRevA.72.042104}{{\em Phys. Rev.~A}
  {\bfseries 72}, 042104 (2005)}.
  
\bibitem{krausbook} K Kraus, {\em States, Effects, and Operations}, Lecture Notes in Physics, Vol. 190, (Editors: K. Kraus, K.
A. B\"ohm, J.D. Dollard, W.H. Wootters), Springer 1983. 

\bibitem{pusey13}
M.~F. {Pusey}, 
  \href{http://dx.doi.org/10.1103/PhysRevA.88.032313}{{\em Phys. Rev.~A}
  {\bfseries 88}, 032313 (2013)}.

\bibitem{skrzypczyk14}
P.~Skrzypczyk, M.~Navascu\'es, and D.~Cavalcanti, 
  \href{http://dx.doi.org/10.1103/PhysRevLett.112.180404}{{\em Phys. Rev.
  Lett.} {\bfseries 112}, 180404 (2014)}.

\bibitem{bowles14}
J.~Bowles, T.~V\'ertesi, M.~T. Quintino, and N.~Brunner, 
  \href{http://dx.doi.org/10.1103/PhysRevLett.112.200402}{{\em Phys. Rev.
  Lett.} {\bfseries 112}, 200402 (2014)}.
  
  

\bibitem{kru} P. Kruszynski, W.M. de Muynck, Math. Phys. {\bf 28}, 1761 (1987).


\bibitem{teiko08}
T.~{Heinosaari}, D.~{Reitzner}, and P.~{Stano}, 
  \href{http://dx.doi.org/10.1007/s10701-008-9256-7}{{\em Foundations of
  Physics} {\bfseries 38}, 1133--1147 (2008)}.

\bibitem{YC}
Y.-C. Liang, R.W. Spekkens, H.M. Wiseman, Phys. Rep. {\bf 506}, 1 (2011).


\bibitem{Note1}
Note that in principle Bob may use POVMs. Since we can restrict to pure
  two-qubit entangled states, it is sufficient to consider for Bob POVM with 4
  outcomes: see G.~{Mauro D'Ariano}, P.~{Lo Presti}, and P.~{Perinotti}, 
  \href{http://dx.doi.org/10.1088/0305-4470/38/26/010}{{\em J. Phys.
  A Math. Gen.} {\bfseries 38}, 5979 (2005)}.
.

\bibitem{navascues12}
M.~{Navascu{\'e}s} and D.~{P{\'e}rez-Garc{\'{\i}}a}, 
  \href{http://dx.doi.org/10.1103/PhysRevLett.109.160405}{{\em Phys. Rev.
  Lett.} {\bfseries 109}, 160405 (2012)}.

\bibitem{tsirelson93}
B.~S. Tsirelson, 
 {\em Hadronic Journal Supplement} {\bfseries 8}, 329--345
  (1993).
  
  \bibitem{tamas}
 T.~ {V{\'e}rtesi} and  K.F.~ {P{\'a}l},  ``{Bounding the dimension of bipartite quantum systems},''
 \href{http://dx.doi.org/10.1103/PhysRevA.79.042106}{{\em Phys. Rev. A} {\bfseries 79}, 042106 
  (2009)}.
  
  \bibitem{collins04}
D.~{Collins} and N.~{Gisin}, ``{A relevant two qubit Bell inequality
  inequivalent to the CHSH inequality},''
  \href{http://dx.doi.org/10.1088/0305-4470/37/5/021}{{\em Journal of Physics A
  Mathematical General} {\bfseries 37}, 1775--1787 (2004)}.

\bibitem{fritz14}
R.~Kunjwal, C.~Heunen, and T.~Fritz, 
  \href{http://dx.doi.org/10.1103/PhysRevA.89.052126}{{\em Phys. Rev. A}
  {\bfseries 89}, 052126 (2014)}.

\bibitem{Uola}
 R. Uola, T. Moroder, O. G\"uhne, 
\href{http://link.aps.org/doi/10.1103/PhysRevLett.113.160403}{{\em Phys. Rev. Lett.}
{\bfseries 113}, 160403(2014)}.


\break
\begin{appendix}
\end{thebibliography}

\providecommand{\href}[2]{#2}\begingroup\raggedright
\section{Tight Bell inequalities when Alice has 3 measurements}

To study the (possible) violation of Bell inequalities with the hollow triangle (8) we consider tight Bell inequalities of form (9). For $n\leq 5$, all of them are known. For $n=3$, there is a single tight Bell inequality (besides CHSH), which is
\begin{align} \label{I3322}
I_{3322}&= -\mean{A_1} - \mean{A_2}+\mean{B_1}+\mean{B_2}  \\
&+ \mean{A_1B_1} + \mean{A_1B_2}+\mean{A_1B_3}  \nonumber  \\ &
+\mean{A_2B_1}+\mean{A_2B_2}-\mean{A_2B_3}  \nonumber  \\ &
 +\mean{A_3B_1}-\mean{A_3B_2} \leq 4 \nonumber
\end{align}
using the notation of equation (10).
For $n=4$, there are 3 new tight inequalities. For $n=5$, we characterized completely the polytope and found a single new tight Bell inequality:
\begin{align}
I_{3522}= &\mean{B_1} + \mean{B_5}  -\mean{A_1B_1} + \mean{A_1B_2}
\\ &-\mean{A_1B_4}+\mean{A_1B_5}+\mean{A_2B_1} \nonumber  \\ &+
 \mean{A_2B_2} +  \mean{A_2B_3}  -\mean{A_2B_5} \nonumber  \\ &+
 \mean{A_3B_1}-\mean{A_3B_3} -\mean{A_3B_4}-\mean{A_3B_5}  \leq 6. \nonumber
\end{align}
As shown in Table~I, none of these inequalities can be violated with the hollow triangle (8).

\section{Planar Hollow Triangle violating the Bell inequality $I_{3322}$}

We show that a simple hollow triangle, featuring only co-planar POVM elements, can violate the Bell inequality $I_{3322}$. Hence, a hollow triangle can lead to Bell inequality violation in the simplest possible scenario, \ie $N=3$.

Consider the projector
\begin{align}
P(\theta)=\frac{1}{2}(\cos(\theta) \sigma_z + \sin(\theta)\sigma_x + \openone),
\end{align}
and the POVM elements defined by
\begin{align}\label{tamas'}
M'''_{0\vert 1}(\eta) &= \eta \frac{2}{3} P(\pi/3) \nonumber \\
M'''_{0\vert 2}(\eta) &= \eta\frac{2}{5} P(0) \nonumber \\
M'''_{0\vert 3}(\eta) &= \eta\frac{2}{3} P(-\pi/3)
\end{align}
and $M_{1|x}'''(\eta)= \openone - M_{0|x}'''(\eta)$, and $0 \leq \eta \leq 1$. For any $\eta$, this set of POVMs is pairwise jointly measurable, hence it cannot violate the CHSH Bell inequality. However, for $\eta=1$, the set is not jointly measurable, as it can be used to violate the Bell inequality $I_{3322}$ \eqref{I3322}, the state and measurements (for Bob) leading to the maximal violation of 4.0595 are given by
\begin{align}
\ket{\psi}=&0.369888\ket{00}+0.316512\ket{01}\nonumber\\
&+0.543175\ket{10}-0.684079\ket{11}
\end{align}
and
\begin{align}
M^B_{0\vert 1}&=P(1.353699) \nonumber\\
M^B_{0\vert 2}&=P(0) \nonumber\\
M^B_{0\vert 3}&=P(-0.598747).
\end{align}

The threshold values of $\eta$ for violation of tight Bell inequalities with $n=4,5$ are given in Table II.

\setcounter{table}{1}
\begin{table}[t! h!] \label{table2}
\begin{tabular}{|c|c|}
\hline
& $ \{ M_{a|x}'''(\eta) \} $ \\
\hline \hline
Pairwise JM (CHSH violation)& 1 \\
\hline
Triplewise JM & 0.7142  \\
\hline \hline
Bell violation ($n=3$): $I_{3322}$  &  0.9375 \\
\hline
Bell violation ($n=4$): $I^1_{3422}$  &  0.9616  \\
\phantom{Bell violation ($n=4$)} $I^2_{3422}$   & 0.8999  \\
\phantom{Bell violation ($n=4$)}$I^3_{3422}$  & 0.9382 \\
\hline
Bell violation ($n=5$): $I_{3522}$   & 0.8999  \\
\hline
\end{tabular}
\caption{Critical $\eta$ for Bell inequality violation with the incompatible POVMs described by equation \eqref{tamas'}.
As in Table I, we determine the smallest value of the
parameter $\eta$, such that the set becomes incompatible (JM), and achieves Bell inequality violation.}
\end{table}

\section{SDP method}

Our task is to compute the maximal quantum violation $\beta^*$ of a two-party Bell inequality defined by the vector of coefficients $c$, given a fixed set of measurements operators ${A_{a|x}}$ for Alice. Let's write $\sigma_{b\vert y} = \tr_B(\rho_{AB}\openone \otimes B_{b|y})$ for the assemblage created on Alice's side by Bob's measurements $\{B_{b|y}\}$ on the state $\rho_{AB}$. With this, we may write the conditional probabilities $p(ab|xy)=\tr(A_{a|x}\sigma_{b\vert y})$ and we have to maximize $\sum c_{a,b,x,y}p(ab|xy)$ for fixed $c_{a,b,x,y}$ and $A_{a|x}$, what may be done with the following SDP 
\begin{equation} 
\begin{aligned}
  \text{given: } &\{A_{a\vert x}\}_{ax},\, \{c_{a,b,x,y}\}_{abxy} \\
&\max{  \sum_{a,b,x,y} c_{a,b,x,y}\tr(A_{a|x}\sigma_{b\vert y})} \\
\text{subject to: }\; &   \sigma_{b|y} \geq 0, \quad \forall b,y \\
& \sum_b\tr(\sigma_{b|y})=1 \quad\forall y \\
& \sum_b \sigma_{b|y}=\sum_b \sigma_{b|y'}, \forall y,y'
\end{aligned}\label{SDP}
\end{equation}
This SDP solves the problem because an assemblage $\{\sigma_{b|y}\}_{by}$ respects the positivity, normalization, and non-signaling constraints respectively given by $ \sigma_{b|y} \geq 0$, $\sum_b\tr(\sigma_{b|y})=1$, and $\sum_b \sigma_{b|y}=\sum_b \sigma_{b|y'}$, if and only if one can find a quantum state $\rho_{AB}$ and a POVM $\{B_{b\vert y}\}_{by}$ such that  $\sigma_{b\vert y} = \tr_B(\rho_{AB} \, \openone \otimes B_{b|y})$. Indeed, it is enough to define 
\begin{align}
  \sigma&:= \sum_b \sigma_{b|y} \\
  B_{b|y}&:=\sqrt{\sigma}^{-1}\sigma_{b|y}^T\sqrt{\sigma}^{-1} \\
  \rho_{AB}&:=\left(\openone\otimes\sqrt{\sigma}\right) \ketbra{\Phi}{\Phi}\left(\openone\otimes\sqrt{\sigma}\right),
\end{align}
where $\sigma_{b|y}^T$ is the transposition of $\sigma_{b|y}$ in the computational basis, $\sqrt{\sigma}$ is the unique positive semidefinite square root of $\sigma$ and $\sqrt{\sigma}^{-1}$ is the inverse of $\sqrt{\sigma}$ on its range (the Moore Penrose inverse) and  $\ket{\Phi}=\sum_i \ket{ii}$ is an (unormalized) maximally entangled state in $\mathbb{C}^d \otimes \mathbb{C}^d$. Direct calculation then shows that $\{B_{b|y}\}_{by}$ is a valid POVM, $\rho_{AB}$ is a quantum state, and that  $\sigma_{b\vert y} = \tr_B(\rho_{AB}\openone \otimes B_{b|y})$. 
 Hence, the SDP presented in Eq.~\eqref{SDP} provides the exact quantum bound of $\beta^*=\max{  \sum_{a,b,x,y} c_{a,b,x,y}\tr(A_{a|x}\sigma_{b\vert y})} $ for a fixed set of Alice's measurements $A_{a|x}$ on the Bell inequality defined by coefficients $\{c_{a,b,x,y}\}_{abxy}$.

%
%

\end{appendix}
\endgroup

\end{document}